\documentclass[twocolumn]{ijns}
\usepackage{graphicx,epsfig,color}
\usepackage[normal]{subfigure}
\usepackage{amsmath,amsthm,amssymb}
\usepackage{floatflt}  
\usepackage{setspace}
\usepackage{rotating,algorithm,algorithmic}
\usepackage{cite}

\def\markboth#1#2{\def\leftmark{#1}\def\rightmark{#2}}

\newtheorem{theorem}{Theorem}

\title{\huge\bf A multi-candidate electronic voting scheme with unlimited participants}
\author{Xi Zhao$^1$, Yong Ding$^2$, Quanyu Zhao$^3$\\
\medskip {\small\it(Corresponding author: Yong Ding)}~\\
{\normalsize  School of Mathematics and Computing Science,
\&  Guilin University of Electronic Technology$^1$}\\
{\normalsize Guilin 541004, China}\\
{\normalsize School of Computer Science and Information Security \& Guilin University of Electronic Technology$^2$}\\
{\normalsize Guilin 541004, China}\\
{\normalsize School of Mathematics and Computing Science \&  Guilin University of Electronic Technology$^3$}\\
{\normalsize Guilin 541004, China}\\
{\normalsize (Email: stone\_dingy@126.com)}\\
}

\date
\newpage
\begin{document}
\maketitle
\thispagestyle{headings}
\begin{abstract}
In this paper a new multi-candidate electronic voting scheme is constructed with unlimited participants.
The main idea is to express a ballot to allow voting for up to $k$ out of the $m$ candidates and unlimited participants.
The purpose of vote is to select more than one winner among $m$ candidates.
Our result is complementary to the result by Sun peiyong$'$ s scheme, in the sense, their scheme is not amenable for large-scale electronic voting due to flaw of ballot structure.
In our scheme the vote is split and hidden, and tallying is made for $G\ddot{o}del$ encoding in decimal base without any trusted third party, and the result does not rely on any traditional cryptography or computational intractable assumption.
Thus the proposed scheme not only solves the problem of ballot structure, but also achieves the security including perfect ballot secrecy, receipt-free, robustness, fairness and dispute-freeness.
\vspace*{0.1cm}
~
\end{abstract}

\section{Introduction}
Electronic voting plays an important role in the cryptographic applications.
In recent years, the traditional voting has been gradually replaced by electronic voting.
An e-voting scheme is a set of protocols that allow a collection of voters to cast their votes, while enabling a collection of authorities to collect the voters, compute the final tally, and communicate the final tally that is checked by tally-clerk.
Generally e-voting protocols adopt the group signature~\cite{mh2009jns}, verifiable password sharing, homomorphic encryption~\cite{cp2011jns} and limited commitment protocol, etc. In 1981, the first e-voting is proposed by Chaum D~\cite{cd1981jns}.
In 1996, Cramer~\cite{rb1996jns} posed the problem of multi-candidate vote. Based on the ElGamal homomorphic encryption, Cramer el al.~\cite{rrb1997jns} proposed a $1-out-of-m$ electronic voting scheme.
However, this scheme is difficult to avoid the situation that same voters repeatedly vote, and it cannot achieve the $k-out-of-m$ e-voting.
In 2001, Damgard el al.~\cite{im2001jns} proposed a $k-out-of-m$ voting scheme based on the Paillier system, while it has the same problem as the Cramer scheme.
In 2006, a multi-candidate e-voting scheme (called Zhong-Huang scheme) was proposed by Zhong et al~\cite{zh2006jns}.
The scheme adopted a secure sum protocol~\cite{kh1980jns, hc1990jns, de1981jns} to achieve the $k-out-of-m$ voting scheme and solve the problem as the Cramer scheme. Nevertheless, by analysis, Zhong-Huang scheme has a flaw that the flaw could disclose the information of voters due to the ballot construction.
To overcome the flaw, in 2012, a secure e-voting scheme with multi-candidate (called Sun-Liu scheme) was proposed by Sun et al~\cite{sl2012jns}.The scheme achieves the perfect ballot secrecy by adopting the random number to blind the vote. However, this scheme is only suitable for small-scale voting due to the flaw of ballot construction.

In this paper, a multi-candidate electronic voting scheme with unlimited participants was proposed by using the $G\ddot{o}del$ encoding to tally vote. Firstly the identity of candidate is masked by using prime number and the process is operated in decimal. Then the ballot of voter and random number are distributed to other voters by secure channel. Finally the result is computed by every voter and is broadcasted to all participants. In this protocol, every voter can tally to achieve the fairness and unlimited participants, perfect ballot secrecy, receipt-free, robustness and dispute-freeness.

In section 2, a brief introduction of $G\ddot{o}del$ coding is introduced and some conditions and properties are defined. The scheme of Sun-Liu is presented in detail and a simple example is given in Section 3. In section 4 the new scheme is presented in detail and the experimental results are shown in section 5. Section 6 gives the performance analysis of the presented scheme. The conclusions are described in Section 7.

\section{Preliminary knowledge}

\subsection{$G\ddot{o}del$ coding}
In 1931,$G\ddot{o}del$ used a system based on prime factorization and assigned a unique natural number to each basic symbol in the formal language of arithmetic. To encode an entire formula, the following system was used by $G\ddot{o}del$.

Given a sequence $(x_1,x_2,\ldots,x_n)$ of positive integers, the $G\ddot{o}del$ encoding of the sequence is the product of the $n$ primes raised to their corresponding values in the sequence: $enc(x_1,x_2,\ldots,x_n)=2^{x_1}3^{x_2}\ldots p_n^{x_n}$.

According to the fundamental theorem of arithmetic, any number (and, in particular, a number obtained in this way) can be uniquely factored into prime factors, so it is possible to recover the original sequence from its $G\ddot{o}del$ numbers (for any given number n of symbol to be encoded).

\subsection{Basic definition}
\begin{description}
\item[Definition 1:]In the protocol, the semi-honest person fully obeys the regulation of e-voting. However the semi-honest person could collect all records and try to infer information of other voters in the acting process. This semi-honest person is also called the passive attacker.In this paper, involved members are assumed to be semi-honest members.

\item[Definition 2:]The e-voting scheme with unlimited participants is called secure and feasible, if it satisfies the following properties.

 (1) Unlimited participants: The e-voting scheme allows for unlimited participants.

 (2) Perfect ballot secrecy: If there are $t(t<n)$ voters have been colluded among $n$ voters, $t$ voters only get their own vote and cannot get the voting result of other voters.

 (3) Receipt-free: Voters are unable to provide his vote to outside world and the vote data cannot be proved to other people.

 (4) Robustness: Voters cannot interrupt the normal process of the protocol.

 (5) Fairness: All voters can get the voting result, and a single individual cannot tally in advance.

 (6) Dispute-freeness: Any voter can verify the correctness of the result and can test the honesty of voters.
\end{description}

\section{Sun - Liu Scheme and Analysis of applicability}

\subsection{ Sun-Liu scheme}

In Sun-Liu scheme, $n$ voters $(P_1,P_2,\ldots,P_n)$ are set with equal status, same rights, and were registered before the vote. There are $m$ candidates $(C_1,C_2,\ldots,C_m)$, and the ballot format is $00\ldots0a_j$. The bit length is $k$, the last bit is $a_j$, and the remaining bits are $0$. $k=[\log_{2}n]+1$, so the total bit length is $mk$. The value of $a_j$ depends on the wish of voters. Only when voters agree with the candidate, $a_j=1$; otherwise $a_j=0$. The voting scheme is divided into three stages: voting, casting, and tallying. The ballot format is shown in Table 1 .
\begin{table}[!hbpt]
\caption{The ballot format to support multi-candidate} \label{t01}
\begin{center}
\begin{tabular}{cc} \hline
  Candidate  & Ballot Format \\ \hline
  $C_1$      & $00\ldots0a_1$\\
  $C_2$      & $00\ldots0a_2$ \\
  $\vdots$   & $\vdots$       \\
  $C_m$      & $00\ldots0a_m$ \\ \hline
\end{tabular}
\end{center}
\end{table}

(1)	Voting

$P_i(i=1,2,\ldots,n)$ makes a vote for candidates $C_j(j=1,2\ldots,m)$ in the electronic vote. The value of $a_j(j=1,2,\ldots,m)$ is $1$ or $0$. Then according to the binary sequence, $P_i(i=1,2,\ldots,n)$ gets a decimal number $v_i$.

(2)	Casting

$P_i(i=1,2,\ldots,n)$ randomly selects a decimal $R_i$ and gets $S_i=v_i+R_i$. Then $S_i$ is divided into $n$ smaller integers $s_{ij}$, such that $S_i=\sum\limits_{i=1}^{n}s_{ij}$. $s_{ij}$ is sent to other voters $P_j (j=1,2,\ldots,n,j\neq i)$ by a secure channel\cite{Wang:2016:BHP:2925426.2926256}\cite{WANG2017219}. $P_j$ calculates $S_i^{'}=\sum\limits_{j=1}^{n}s_{ij}$ after receiving the $n-1$ random numbers $s_{ij}(i=1,2,\ldots,n,i\neq j)$.Then $P_i(i=1,2,\ldots,n)$ calculates $v_i^{'}=S_i^{'}-R_i$.

(3)	Tallying

$v_i^{'}$ is broadcasted to the other voters. Then $P_i(i=1,2,\ldots,n)$ calculates the sum $T=\sum\limits_{i=1}^{n}v_i^{'}$. $T$ from:

$T=\sum\limits_{i=1}^{n}v_i{'}=\sum\limits_{i=1}^{n}(S_i^{'}-R_i)=\sum\limits_{i=1}^{n}S_i^{'}-\sum\limits_{i=1}^{n}R_i
=\sum\limits_{i=1}^{n}\sum\limits_{j=1}^{n}s{ij}-\sum\limits_{i=1}^{n}R_i=
\sum\limits_{j=1}^{n}\sum\limits_{j=1}^{n}s{ij}-\sum\limits_{i=1}^{n}R_i=\sum\limits_{i=1}^{n}S_i-\sum\limits_{i=1}^{n}R_i
=\sum\limits_{i=1}^{n}(S_i-R_i)=\sum\limits_{i=1}^{n}v_i$.

Finally, $P_i(i=1,2,\ldots,n)$ is able to get the binary sequence from $T$,then the binary sequence is intercepted in every $k$ bits. The voting result of $C_j(j=1,2,\ldots,m)$ is obtained.
\subsection{ Example: 7 voters vote for 3 candidates}

The example of 7 voters vote for 3 candidates is used to explain the Sun-Liu scheme. The scheme is divided into 3 steps.

(1)	Voting

The voting result is shown in Table 2

\begin{table}[!hbpt]
\caption{7 ballots for 3 candidates} \label{t02}
\begin{center}
\begin{tabular}{cccccc} \hline

   &\multicolumn{3}{c}{Candidate} &\multicolumn{2}{c}{Vote} \\ \cline{2-6}
  Voter   &A   &B  &C  &Binary    &Decimal       \\ \hline
  $P_1$   &001&000&001 &001000001 &65 \\
  $P_2$   &000&000&001 &000000001 &1 \\
  $P_3$   &001&000&000 &001000000 &64   \\
  $P_4$   &000&001&001 &000001001 &9 \\
  $P_5$   &001&001&001 &001001001 &73    \\
  $P_6$   &000&001&000 &000001000 &8     \\
  $P_7$   &000&001&000 &000001000 &8  \\ \hline
\end{tabular}
\end{center}
\end{table}

(2)	Casting

Step 1: Selecting a random number $R_i$ and calculate $S_i$ from:
$$S_i=v_i+R_i$$

The result is shown in Table 3
\begin{table}[!hbpt]
\caption{The result of ballots randomization} \label{t03}
\begin{center}
\begin{tabular}{cccccccc} \hline

          &$P_1$&$P_2$&$P_3$&$P_4$&$P_5$&$P_6$&$P_7$      \\ \hline
  $v_i$   &65   &1    &64   &9    &73   &8    &8 \\
  $R_i$   &10   &12   &2    &6    &11   &25   &6 \\
  $S_i$   &75   &13   &66   &15   &84   &33   &14 \\ \hline
\end{tabular}
\end{center}
\end{table}

Step 2: $S_i$ is divided into seven smaller sub-digits by $P_i (i=1,2,\ldots,7)$ and is sent to other voters. The transfer matrix is shown in Figure.1

\begin{figure}[!hbpt]
\begin{center}
  \includegraphics[width=40mm,height=35mm]{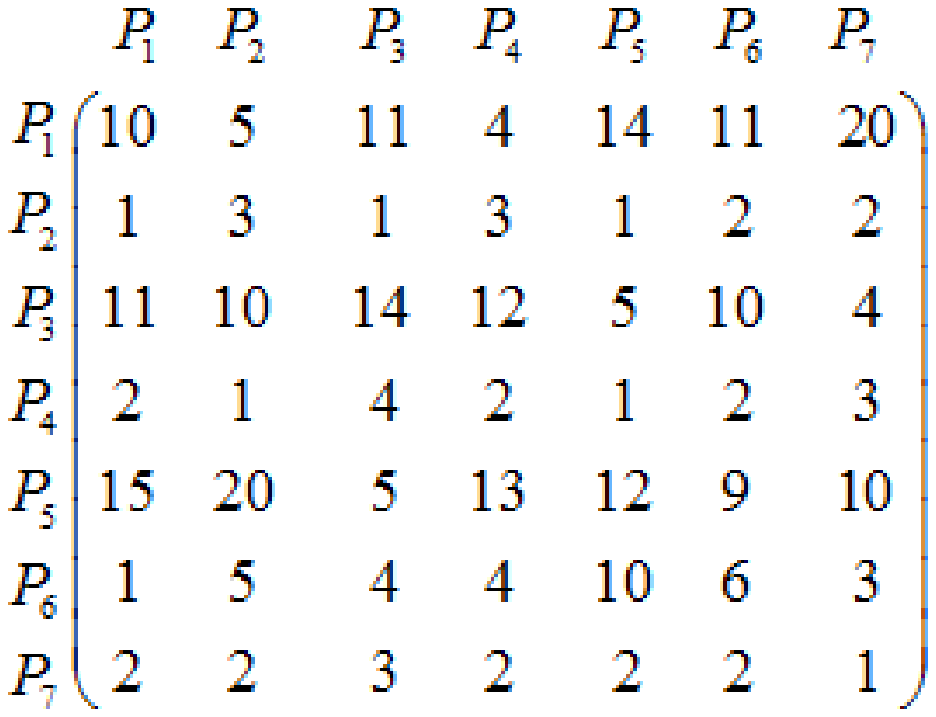}
  \caption{ The transfer matrix of Sun-Liu scheme}
  \label{f001}
\end{center}
\end{figure}

Step 3: Calculate the sum of $j$th list from:
$$S_i^{'}=\sum\limits_{j=1}^{n}s_{ij} $$

Here $s_{ij}$ is an element of $S_i$ divided in Figure.1.

Step 4: Calculate the result of derandomization from:
$$v_i^{'}=S_i^{'}-R_j$$

Step 5: Broadcast $v_i^{'}$, the result is shown in Table 4.

\begin{table}[!hbpt]
\caption{Voter cuts the random number } \label{t04}
\begin{center}
\begin{tabular}{cccccccc} \hline

              &$P_1$&$P_2$&$P_3$&$P_4$&$P_5$&$P_6$&$P_7$      \\ \hline
  $S_i^{'}$   &42   &46   &42   &40   &45   &42   &43 \\
  $R_i$       &10   &12   &2    &6    &11   &25   &6 \\
  $v_i^{'}$   &32   &34   &40   &34   &34   &17   &37 \\ \hline
\end{tabular}
\end{center}
\end{table}

(3)	Tallying

he amount of votes is $T=\sum\limits_{i=1}^{7}v_i^{'}=32+34+40+40+38+12+37=228$,then sequence $011100100$is got from 228 and is intercepted in every 3 bits. So the voting result of A, B, and C are 3, 4, and 4 respectively.

\subsection{Adaptive analysis of Sun-Liu scheme}

In Sun-Liu scheme, the bit length of vote is $mk=m(\log_2n+1)$, when $mk>64$, the bit length will lead to more consumption of computing time. So in $mk>64$, the method is low efficiency. If $m=9,n=1000$, the binary bit length is $mk=90$, the situation has more  consumption of computing time than $mk<64$.

\section{A multi-candidate electronic voting scheme with unlimited participants}

In view of the boundedness of the Sun-Liu scheme, the ballot structure is improved to achieve the aim of unlimited participants by the multiplication in the finite field $F_p$ ($p$ is a big prime). The process is following:

(1)	Voting

$P_1,P_2,\ldots,P_m$ (a total of $m$ primes) are respectively $m$ identities of candidates. Voter $L_i(i=1,2,\ldots,n)$
makes vote for $m$ candidates. If the candidate is agreed, $L_i(i=1,2,\ldots,n)$ selects the prime which represents the identity of candidate, otherwise the number ��1�� is selected instead.

(2)	Casting

$L_i(i=1,2,\ldots,n)$ randomly selects arbitrary primes to be regarded random number $R_i$ in the $m$ primes. $L_i(i=1,2,\ldots,n)$ sends $s_{ij}$ to other voters $L_j(j=1,2,\ldots,n,j\neq i)$ by secure channel, where $s_{ij}$ consists of the selected prime, $R_i$, and $1$. $L_j(j=1,2,\ldots,n)$ receives $s_{ij}(i=1,2,\ldots,n,i\neq j)$ from $n-1$ voters, and the process could be showed by transfer matrix Figure.2.

\begin{figure}[!hbpt]
\begin{center}
  \includegraphics[width=40mm,height=35mm]{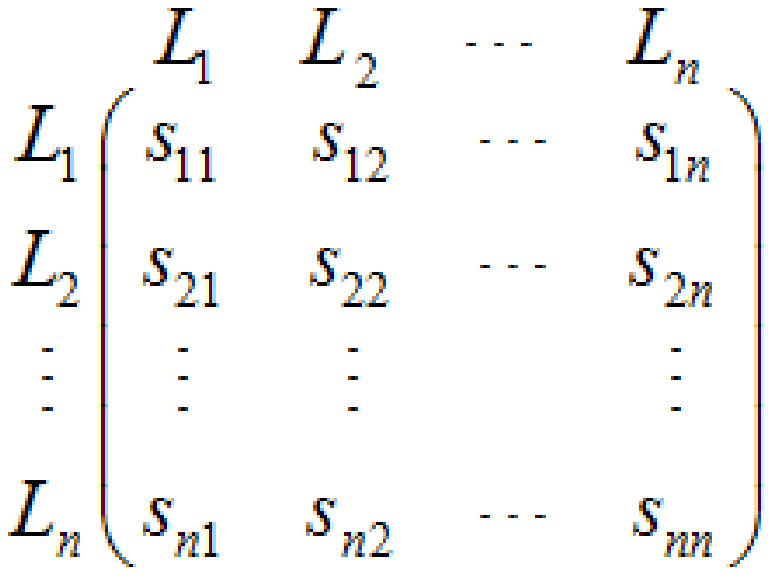}
  \caption{ Transfer matrix}
  \label{f002}
\end{center}
\end{figure}

(3)	Tallying

$L_i(i=1,2,\ldots,n)$ gets the result $v_i=2^{x_1^{'}}3^{x_2^{'}}\dots m^{x_n^{'}}R_i^{-1}$  by calculating ($R_i^{-1}$ is inverse element of $R_i$) and broadcast to $L_j(j=1,2,\ldots,n,j\neq i)$ by secure channel. $L_j$ calculates the product $T$ after receiving other $n-1$ results, namely $T=\prod_{i=1}^{n}v_i=2^{x_1}3^{x_2}\ldots f^{x_n}$, the sequence $(x_1,x_2,\ldots,x_n)$ is the voting result of corresponding candidate respectively.

\section{Example}

\subsection{Example 1}

7 voters vote for 3 candidates

(1)	Voting

The situation of vote is shown in Table 5

\begin{table}[!hbpt]
\caption{7 ballots for 3 candidates} \label{t05}
\begin{center}
\begin{tabular}{ccccc} \hline

   &$P_1$&$P_2$&$P_3$ &Select \\ \cline{2-4}
  Voter   &2   &3  &5   &Random number       \\ \hline
  $L_1$   &2   &1  &5  &2,3 \\
  $L_2$   &1   &1  &5  &3,5 \\
  $L_3$   &2   &1  &1  &5   \\
  $L_4$   &1   &3  &5  &3 \\
  $L_5$   &2   &3  &5  &2,5    \\
  $L_6$   &1   &3  &1  &2,3     \\
  $L_7$   &1   &3  &1  &2,3  \\ \hline
\end{tabular}
\end{center}
\end{table}

(2)	Casting

Step 1: $L_i (i=1,2,\ldots,7)$ gets data $s_{ij}$ and broadcasts to other voters $L_j (j=1,2,\ldots ,n,j\neq i)$ by secure channel. The transfer matrix is shown in Figure.3.

\begin{figure}[!hbpt]
\begin{center}
  \includegraphics[width=40mm,height=30mm]{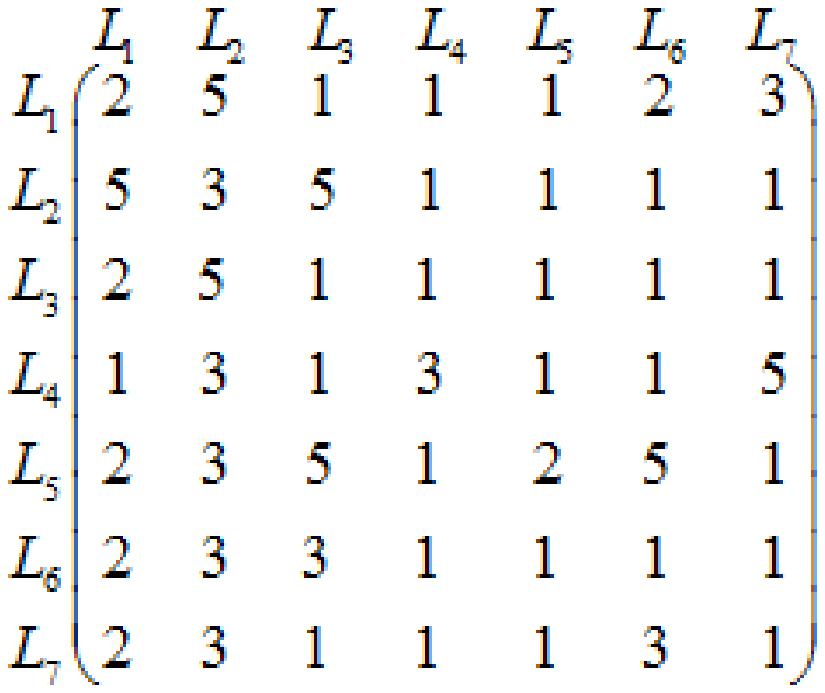}
  \caption{The transfer matrix}
  \label{f003}
\end{center}
\end{figure}

Step 2: After calculating the product of $ith$ list and broadcasts to other voters, the result is :

$$v_1=2^4\cdot5^1\cdot3^{-1}$$
$$v_2=3^4\cdot5^1$$
$$v_3=3^1\cdot5^1$$
$$v_4=1$$
$$v_5=5^{-1}$$
$$v_6=5^1$$
$$v_7=5^1\cdot2^{-1}$$

(3)	Tallying

The amount of vote is $T=\prod_{i=1}^{n}v_i=v_1v_2v_3v_4v_5v_6v_7=2^3\cdot3^4\cdot5^4$, so the voting result of $P_1,P_2,P_3$ are 3, 4, 4 respectively.

\subsection{Example 2}

1000 voters vote for 9 candidates

(1)	Voting

The voting situation is shown in Table 6

\begin{table}[!hbpt]
\caption{ 1000 voters vote for 8 candidates} \label{t06}
\begin{center}
\begin{tabular}{cccccccccc} \hline

          &$P_1$&$P_2$&$P_3$&$P_4$&$P_5$&$P_6$&$P_7$&Select \\ \cline{2-8}
  Voter   &2    &3    &5    &7    &11   &13   &17   &Random number   \\ \hline
  $L_1$   &2    &1    &1    &7    &1    &13   &1    &2,3,5\\
  $L_2$   &1    &1    &5    &7    &1    &1    &1    &3,7,17\\
  \vdots  &\vdots&\vdots&\vdots&\vdots&\vdots&\vdots&\vdots&\vdots   \\
  $L_4$   &1    &1    &1    &7    &1    &1    &17   &7,11,13,17\\ \hline

\end{tabular}
\end{center}
\end{table}

(2)	Casting

Step 1: $L_i (i=1,2,\ldots,1000)$ gets $s_{ij}$ and sends it to $L_j (j=1,2,\ldots,n,j\neq i)$ by secure channel, the transfer matrix is shown in Figure.4.

\begin{figure}[!hbpt]
\begin{center}
  \includegraphics[width=40mm,height=30mm]{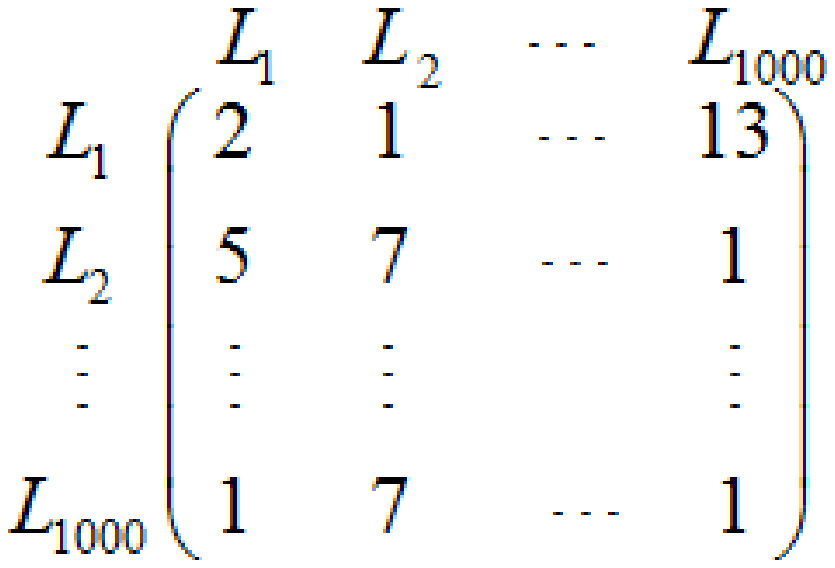}
  \caption{The transfer matrix}
  \label{f005}
\end{center}
\end{figure}

Step 2: After calculating the $ith$ list, $L_i$ gets the result $v_i$ and broadcasts by checking table 4, the result is following:
$$L_1=2^1\cdot 5^1\cdots 2^{-1}\cdot3^{-1}\cdot5^{-1}$$
$$L_2=7^2 \cdots 3^{-1}\cdot 7^{-1}\cdot17^{-1}$$
$$\vdots$$
$$L_1000=13^1\cdots 7^{-1}\cdot11^{-1}\cdot13^{-1}\cdot17^{-1}$$

(3)	Tallying

The amount of vote is $T=\prod_{i=1}^{n}v_i=v_1v_2\cdots v_{1000}=2^{x_1}\cdot3^{x_2}\cdot5^{x_3}\cdot7^{x_4}\cdot11^{x_5}\cdot13^{x_6}\cdot17^{x_7}$, so the vote of $P_1 P_2\cdots P_7 $ are $x_1 x_2\cdots x_7 $ respectively.

According to the analysis of example, the aim of unlimited participants is achieved in the protocol.

\section{Analysis of performance}

\begin{theorem}
 The result is correct based on $G\ddot{o}del$ encoding.~\label{th1}
\end{theorem}
\begin{proof}
According to Figure.2, a transfer matrix is got by randomly splitting the ballot. $\prod_{i=1}^{n}\prod_{j=1}^{n}s_{ij}=\prod_{j=1}^{n}\prod_{i=1}^{n}s_{ij}=T$,so the product of all elements does not change before sending vote and after receiving vote. The $G\ddot{o}del$ encoding sets up a sole corresponding relation for $T$ and limited positive integer sequence. According to the relation, the voting result is got. So the result is right.
\end{proof}

\begin{theorem}
 If all voters are semi-honest, the e-voting scheme with unlimited participants is secure and feasible.~\label{th2}
\end{theorem}
\begin{proof}
The following ideas reflectively prove that our program meets all security properties raised by defination 2

(1) Unlimited participants: The ballot format adopts the form of decimalism and it has not relation with the amount of participants. In theory, the amount of candidates only relates to the number of primes. For example, if the election has 1 billion candidates, 1 billion primes are found. Then according to the protocol, the voting result is got.

(2) Perfect ballot secrecy: $L_i (i=1,2,\ldots,n)$ randomly selects the random number $R_i$ to conceal the ballot before sending the vote. Only when $n-1$ voters are collaborated, the secret ballot of voter is calculated by the protocol. So the protocol can resist the attack from $n-1$ passive attackers. When $t(t<n)$  voters are collaborated, the voters only can get their own ballot and cannot get other voting result. So perfect ballot secrecy can be achieved.

(3) Receipt-free: By the random number, the ballot is concealed and $s_{ij}$ is split into n parts and is sent to other voters. Any voter cannot reconstruct the complete vote, because $L_j$ only can receive $s_{ij}$ from $L_i$. So the vote is receipt-free.

(4) Robustness: When $t(t<n)$  has been abstained from election, the vote data cannot be sent in casting. In tallying, the process still can be started by cooperation of remaining voters.

(5) Dispute-freeness: The voting result is tallied by all voters, if the result of one or a few voters is not same with most voters, these voters are dishonest persons. Moreover, by making public all $v_i^{'}$ we also can test and verify and track the dishonest voters.

\end{proof}

\subsection{ Scheme comparison}

The scheme need not using binary sequence to express vote, which avoids the situation of digit overflow, at the same time the ballot structure does not have relation with the amount participants. So the protocol can achieve the situation of unlimited participants. Our solution only needs to send selected prime number and random number to $L_j (j=1,2,\ldots ,n,j \neq i)$, then $L_i (i=1,2,\ldots,n)$ gets the result by arranging the data of the $ith$ list and broadcast to other voters. In contrast with the Sun-Liu scheme, our solution not only satisfies the security of e-voting but also improves computing efficiency to achieve unlimited participants.

\section{Conclusion}

The protocol uses the $G\ddot{o}del$ coding to solve the problem of low efficiency in previous schemes. In this paper, the scheme does not rely on traditional encryption algorithms to ensure the security of ballot and does not contain any scabrous computational problem. In information theory, the protocol is secure. The e-voting can be applied to election of any scale by the scheme, for example large-scale electronic auction, electronic tending, et al.
	Although the problem for number of restrictions is solved, the insecure channel and dishonest voter need to be resolved. The voter selects secret random numbers need to be proved legal and the security of electoral package need to be further researched. It is possible (and is part of work in progress) that extensions of our ideas can give a scheme with insecure channel and dishonest voter while maintaing security, robustness and efficiency.

\section*{Acknowledgments}
This study was supported by the National National Science Foundation of China (Grant Nos.61363069 ), The Innovation Project of GUET Graduate Education (2016YJCX44) and Guangxi Key Laboratory of Cryptography and Information Security.

We thank yining Liu, Huiyong Wang and Lei Li for helpful comments and discussions.

\addcontentsline{toc}{chapter}{\protect\numberline{}{REFERENCES}}
\bibliography{ijns}
\bibliographystyle{unsrt}

\noindent {\bf Xi Zhao} is currently pursuing his M.S. degree in Guilin University of Electronic Technology, Guilin China. His current research interests include information security and image security.
\\

\noindent {\bf Yong Ding } as born in Chongqing, China. He received the
Ph.D. degrees from the School of Communication Engineering,
Xidian University, Shaanxi, in 2005 . He is currently professor
in Guilin University of Electronic Technology, China. His
current research interests include cryptography and information
security.\\

\noindent {\bf Quanyu Zhao} is currently pursuing his M.S. degree
in Guilin University of Electronic Technology, Guilin,
China. He received the B.S. degree in Mathematics and
Applied Mathematics from Huainan Normal University,
Anhui, China, in 2014. His research interests focus on
e-voting, micropayment, e-lottery, group key transfer,
and oblivious transfer.

\end{document}